\documentclass{article}

\usepackage{arxiv}
\usepackage[utf8]{inputenc}     % allow utf-8 input
\usepackage[T1]{fontenc}        % use 8-bit T1 fonts
\usepackage[
  colorlinks=true,
  unicode=true,                 % or unicode
]{hyperref}
\usepackage{url}                % simple URL typesetting
\usepackage{booktabs}           % professional-quality tables
\usepackage{amsmath}			% AMS packages
\usepackage{amssymb}				
\usepackage{amsthm}
\usepackage{stmaryrd}
\usepackage{pict2e}
\usepackage{appendix}			% appendicies
\usepackage{nicefrac}           % compact symbols for 1/2, etc.
\usepackage{microtype}          % microtypography
\usepackage[svgnames]{xcolor}
\usepackage{extarrows}
\usepackage{pb-diagram}
\usepackage{textgreek}
\usepackage{bm}
\usepackage{enumitem}
\setlist{nosep}
\setlist[1]{labelindent=\parindent} 
\setlist[enumerate,1]{label={\rm(\arabic*)}, ref={\rm(\arabic*)}}
\setlist[itemize,1]{label=--}
\setlist[description]{font=\sffamily\bfseries}
\usepackage{lipsum}
\theoremstyle{plain}
\newtheorem{proposition}{Proposition}
\newtheorem{lemma}{Lemma}
\newtheorem{theorem}{Theorem}
\newtheorem{corollary}{Corollary}
\theoremstyle{definition}
\newtheorem{definition}{Definition}

\theoremstyle{remark}
\newtheorem{remark}{Remark}

\def\bs#1{\boldsymbol{#1}}
\def\nil{\bs\epsilon}
\newcommand{\slice}[3][]{_{\ifx\\#1\\\else\strut\fi#2\mathbin{..}#3}}

\newcommand{\AT}[1]{\textit{\bfseries#1}}
\newcommand{\Cat}[1]{\mathbb#1}
\newcommand{\onm}[1]{\mathsf{#1}}
\newcommand{\cat}[1]{\bm{\mathsf{#1}}}
\DeclareMathOperator{\dom}{dom}
\DeclareMathOperator{\cod}{cod}
\newcommand{\id}[1][]{\mathop{\mathrm{id}\ifx\\#1\\\else_{#1}\fi}}
\newcommand{\Fnr}[1]{\mathop{\boldsymbol#1\vphantom{|}}}
\newcommand{\fnr}[2][]{\mathop{\mathsf{#2\vphantom{|}}\ifx\\#1\\\else_{#1}\fi}}
\newcommand{\crr}[1]{\underline{#1}}

\newcommand{\tuple}[2]{\langle#1,#2\rangle}
\newcommand{\bigtuple}[2]{\big\langle#1,#2\big\rangle}

\newcommand{\lbparen}{%
  \mathopen{%
    \sbox0{$()$}%
    \setlength{\unitlength}{\dimexpr\ht0+\dp0}%
    \raisebox{-\dp0}{%
      \begin{picture}(.32,1)
      \linethickness{\fontdimen8\textfont3}
      \roundcap
      \put(0,0){\raisebox{\depth}{$($}}
      \polyline(0.32,0)(0,0)(0,1)(0.32,1)
      \end{picture}%
    }%
  }%
}
\newcommand{\rbparen}{%
  \mathclose{%
    \sbox0{$()$}%
    \setlength{\unitlength}{\dimexpr\ht0+\dp0}%
    \raisebox{-\dp0}{%
      \begin{picture}(.32,1)
      \linethickness{\fontdimen8\textfont3}
      \roundcap
      \put(-0.08,0){\raisebox{\depth}{$)$}}
      \polyline(0,0)(0.32,0)(0.32,1)(0,1)
      \end{picture}%
    }%
  }%
}
\newcommand{\ana}[1]{\lbparen#1\rbparen}
\newcommand{\pows}[2][]{\mathcal P\ifx\\#1\\\else_+\fi(#2)}
\newcommand{\N}[1][]{\mathbb N\ifx\\#1\\\else_{#1}\fi}
\newcommand{\Ntf}[1][]{\mathbf{N}\ifx\\#1\\\else\hspace{-1pt}\onm #1\fi}

\newcommand{\ext}[1]{\id[\onm1]+#1}
\newcommand{\tr}[1]{#1_\mathrm{tr}}
\newcommand{\out}[1]{#1_\mathrm{out}}
\newcommand{\pr}[1][]{\mathop{\mathrm{pr}\vphantom{|}}\ifx\\#1\\\else_{#1}\fi}
\def\fault{\mathop{\Downarrow}}
\renewcommand{\S}[1][]{\fnr{S_{\Ntf[#1]}}}
\newcommand{\D}[1][]{\fnr{D_{\Ntf[#1]}}}
\newcommand{\dr}[1]{\mathfrak{#1}}
\newcommand{\Sys}[1]{\cat{Sys}(#1)}
\newcommand{\fic}[1]{\text{\bf\textnu}\!#1}
\newcommand{\Nxt}[1][]{(\mathbf1+\fnr{I\!d}_\cat{Set})^\Ntf{\ifx\\#1\\\else#1\fi}}

\newcommand{\app}[2]{\mathop{\strut#1}#2}
\newcommand{\labs}[2]{\lambda\ #1\mathbin{.}#2}

\newcommand{\constr}[2]{\llbracket#1\rrbracket\ifx\\#2\\\else_{#2}\fi}

\title{Understanding Safety Constraints Coalgebraically}

\author{
    Grygoriy Zholtkevych\thanks{%
        The author thanks professors R. de Simone, F. Mallet, and L. Liquori for a detailed discussion of the problems related to this paper during his internship at Inria Sophia Antipolis - M\'{e}di and Campus France for funding this internship.} \\
    Department of Theoretical and Applied Computer Science\\
    School of Mathematics and Computer Science \\
    V.N. Karazin Kharkiv National University\\
    4, Svobody Sqr., Kharkiv, 61022, Ukraine \\
    \texttt{g.zholtkevych@karazin.ua} \\
    \And
    Maksym Labzhaniia \\
    Department of Theoretical and Applied Computer Science\\
    School of Mathematics and Computer Science \\
    V.N. Karazin Kharkiv National University\\
    4, Svobody Sqr., Kharkiv, 61022, Ukraine \\
    \texttt{m.labzhaniia@gmail.com}
}

\begin{document}
\maketitle

\begin{abstract}
Safety constraints are crucial to the development of mission-critical systems.
The practice of developing software for systems of this type requires reliable methods for identifying and analysing project artefacts.
This paper proposes a coalgebraic approach to understanding behavioural constraints for systems of a kind.
The advantage of the proposed approach is that it gives a framework for providing abstract semantic models of the domain-specific languages designed for specifying behavioural constraints.

% keywords can be removed
\end{abstract}
%%https://www.overleaf.com/project/5dc48b2db8bdf00001d7eb33
%%

%%
\section{Introduction}
A lot of modern technical systems are compound and smart.
Moreover, they are hybrid in the sense that ones consisted of both physical and cybernetic (software) components.
In other words, we can state that modern technical systems of such a kind are cyber-physical systems (see, for example, \cite{bib:lee}).
This requires the corresponding approaches to designing these systems.
First of all, we need to remark that software complex of such a system contains necessarily reactive components i.e. programs intending rather for providing the required behaviour of the system than for handling data \cite{bib:mana-phueli}.
This is because of the incorrect behaviour of a complex technical system can have serious and some times catastrophic consequences for the system surroundings.
Thus, we should classify such systems as safety-critical \cite{bib:bowen-stavridou}.
As well-known, specification and analysis of the behavioural requirements is the most critical phase under safety-critical systems development process \cite{bib:grant} taking into account the fact that a most of system faults and errors are consequences of incorrect specifications or of incomplete analysis.
Hence, we need a dependable and mathematically grounded toolkit for specifying and analysing behavioural constraints for designing such systems.

This paper introduces some general framework for constructing rigorous models of safety constraints, which can be used for constructing domain-specific languages for specifying safety constraints.
This framework can be included as a component of the mentioned above development toolkit.

In the paper, we use coalgebraic techniques as the main research methodology.
This choice is motivated by the fact that universal coalgebras give an adequate mathematical tool for modelling behaviour of systems \cite{bib:rutten,bib:jacobs-2017}.

Sec.~\ref{sec:notation} introduces the needed notions and notation for sequences of system notifications.

For completeness and reader convenience, we have collected the used coalgebraic concepts in Sec.~\ref{sec:preliminaries}.

Sec.~\ref{sec:systems} is devoted to studying concrete endofunctors and properties of the coalgebras corresponding to these endofunctors related to safety constraints.

Finally, Sec.~\ref{sec:constraints}
is the central section of the paper.
It contains definitions of the target constructions and results providing achieving of the claimed goals. 
%%
%% General Preliminaries
\section{Basic Concepts and Notation}\label{sec:notation}% file 12.tex
In this section, we assume that $X$ is a finite set with at least two elements.
Elements of this set are usually interpreted as system notifications.

A mapping (generally speaking partial) $\bs{s}:\N\to X$ is called an $X$-\AT{sequence} if for any $k\in\N$, $\app{\bs{s}}{k'}$ is defined whenever $\app{\bs{s}}{k}$ is defined and $0\leq k'<k$.

An $X$-sequence $\bs s$ is called an $X$-\AT{word} if there exists $k\in\N$ such that $\app{\bs{s}}{k}$ is undefined; in contrast, an $X$-sequence $\bs s$ is called a $X$-\AT{stream} if $\app{\bs{s}}{k}$ is defined for all $k\in\N$.

We use the notation $X^\ast$ for referring to the set of all $X$-words, and $X^{\N}$ for referring to the set of all $X$-streams.
The set $X^\ast$ contains the sequence defined nowhere, which is below denoted by $\nil$.

We use also the notation $X^\infty$ for referring to the set $X^\ast\bigcup X^{\N}$, and $X^+$ for referring to the set of all $X$-words without the word defined nowhere.

For an $X$-word $\bs u$, we denote by $|\bs u|$ the minimal natural number such that $\app{\bs{u}}{|\bs u|}$ is undefined.
\\
The value $|\bs u|$ where $\bs{u}\in X^\ast$ is called \AT{length} of $\bs u$.

As usually, we identify $n\in X$ with the $X$-word $\bs{u}\in X^\ast$ such that $\app{\bs u}{0}=n$ and $\app{\bs u}{k}$ is undefined if $k>0$.

For $\bs{u}\in X^\ast$ and $\bs{s}\in X^\infty$, we denote by $\bs{us}$ the next $X$-sequence
\[
    \bs{us}=\labs{k\in\N}{\left\lbrace\begin{array}{cc}
        \app{\bs u}{k} & \text{if }k<|\bs u| \\
        \app{\bs s}{(k-|\bs u|)} & \text{otherwise}
    \end{array}\right.}
\]

Below we need in the following set
\[
    n^{-1}\cdot A=\{\bs{u}\in X^\ast\mid n\bs{u}\in A\}\qquad\text{where }A\subset X^\ast\text{ and }n\in X.
\]

For $\bs{s}\in  X^\infty$ and $m\in\N$, we denote by $\bs{s}\slice{m}{}$ the next $X$-sequence
\[
    \bs{s}\slice{m}{}=\labs{k\in\N}{\left\lbrace\begin{array}{cc}
        \app{\bs s}{(k+m)} & \text{if this value is defined} \\
        \text{ is undefined} & \text{otherwise}
    \end{array}\right.}
\]

Also for $\bs{s}\in X^\infty$ and $m,l\in\N$, we denote by $\bs{s}\slice{m}{l}$ the next $X$-word
\[
    \bs{s}\slice{m}{l}=\labs{k\in\N}{\left\lbrace\begin{array}{cc}
        \app{\bs s}{(k+m)} & \text{if this value is defined and }k<l-m \\
        \text{ is undefined} & \text{otherwise}
    \end{array}\right.}
\]

The principal concepts for our studying is given by the following definitions
\begin{definition}
A subset $P\subset X^\ast$ is called \AT{prefix-free} if $\bs{u}\in P$ ensures $\bs{u}\slice{0}{m}\notin P$ whenever $0\leq m<|\bs u|$.
\end{definition}
\begin{remark}
If a prefix-free subset of $X^\ast$ contains $\nil$ then this subset is $\{\nil\}$.
Indeed, if a prefix-free subset of $X^\ast$ contains both $\nil$ and another $X$-word $\bs{u}$ then $\bs{u}\slice{0}{0}=\nil$ cannot belong to this subset.
This contradiction grounds the remark.
\end{remark}
\begin{definition}[see \cite{bib:alpern-schneider}]
A \AT{safety constraint} is a subset $S\subset X^{\N}$ such that $\bs{s}\in S$ if for any $m\in\N$, $\bs{s}\slice[*]{0}{m}=\bs{s}'\slice[*]{0}{m}$ for some $\bs{s}'\in S$.
\end{definition}

%%
%% Coalgebras Preliminaries
\section{Coalgebras Preliminaries}\label{sec:preliminaries}% file 13.tex
In this section, we remind the basic definitions and facts related to the concept of a coalgebra in an arbitrary category.
In addition, we give some specific concepts in the case when $\Cat{C}=\cat{Set}$.

Thus, we assume the category $\Cat{C}$ and the endofunctor $\Fnr{F}$ of $\Cat{C}$ are given and held fixed in this section\footnote{%
General information on category theory can be found in \cite{bib:maclane,bib:awodey}
}.
\begin{definition}\label{def:coalgebra}
A morphism $a$ of $\Cat{C}$ is called an $\Fnr{F}$-\AT{coalgebra} if $\cod a=\Fnr{F}(\dom a)$.
In the case, $\dom a$ is called the \AT{carrier} of $a$ and denoted below by $\crr a$.
\end{definition}
\begin{definition}\label{def:coalgebra-morphism}
Let $a$ and $b$ be $\Fnr{F}$-coalgebras then a morphism $f\colon\crr a\to\crr b$ is called an $\Fnr{F}$-\AT{morphism} from $a$ into $b$ (symbolically, $f:a\to b$) if the diagram
\[\begin{diagram}
    \node{\crr a}
        \arrow{e,t}{f}
        \arrow{s,l}{a}
    \node{\crr b}
        \arrow{s,r}{b}\\
    \node{\Fnr{F}\crr a}
        \arrow{e,t}{\Fnr{F}f}
    \node{\Fnr{F}\crr b}
\end{diagram}
\quad\text{commutes or, equivalently, the equation}\quad\app{(\Fnr{F}f)}{a}=\app{b}{f}\quad\text{holds}.
\]
\end{definition}
\begin{proposition}
The class of $\Fnr{F}$-coalgebras equipped with $\Fnr{F}$-morphisms is a category denoted usually by $\cat{Coalg}_{\Fnr{F}}(\Cat{C})$ or $\cat{Coalg}_{\Fnr{F}}$ if the category $\Cat{C}$ is clear from the context.
\end{proposition}
\begin{definition}\phantom{x}\par
\begin{enumerate}
\item 
A terminal object of $\cat{Coalg}_{\Fnr{F}}$ if it exists is called a \AT{final} $\Fnr{F}$-coalgebra, which is denoted by $\fic{\Fnr{F}}$.
\item
For any $\Fnr{F}$-coalgebra $a$, the unique $\Fnr{F}$-morphism from $a$ into $\fic{\Fnr{F}}$ is called an \AT{anamorphism} and denoted by $\ana{a}$\,.
\end{enumerate}
\end{definition}
The concept of bisimulation is a key concept in the theory of universal coalgebras.
\begin{definition}[see P. Aczel and N. Mendler \cite{bib:aczel-mendler}]\label{def:bisimulation}
A bisimulation of $\Fnr{F}$-coalgebras $a$ and $b$ is a span $a\xleftarrow{r_a}r\xrightarrow{r_b}b$ in category $\cat{Coalg}_{\Fnr{F}}(\Cat{C})$ such that the span $\crr a\xleftarrow{r_a}\crr r\xrightarrow{r_b}\crr b$ in category $\Cat{C}$ is a mono-span\footnote{%
A span $X\xleftarrow{f}R\xrightarrow{g}Y$ in an arbitrary category is a mono-span if for any morphisms $h',h'':Z\to R$, equations $fh'=fh''$ and $gh'=gh''$ ensure $h'=h''$.}.
\end{definition}
The next proposition demonstrates that coalgebra morphisms give the simplest examples of bisimulations.
\begin{proposition}
For any $\Fnr{F}$-coalgebras $a$ and $b$ and $\Fnr{F}$-morphism $f:a\to b$, the span $a\xleftarrow{\id[a]}a\xrightarrow{f}b$ is a bisimulation of $a$ and $b$.
\end{proposition}
\begin{proof}
We need to note only that the span $\crr a\xleftarrow{\id[\crr a]}\crr a\xrightarrow{f}\crr b$ is evidently a mono-span.
\end{proof}
\begin{proposition}
Assume that category $\Cat{C}$ is finitely complete and there is a final $\Fnr{F}$-coalgebra then for any bisimulation $a\xleftarrow{r_a}r\xrightarrow{r_b}b$ of $\Fnr{F}$-coalgebras $a$ and $b$, there exists a unique monomorphism $f:\crr r\to P$ such that $r_a=p_af$ and $r_b=p_bf$ where $\crr a\xleftarrow{p_a}P\xrightarrow{p_b}\crr b$ is the pullback of the cospan $\crr a\xrightarrow{\ana{a}}\crr{\fic{\Fnr{F}}}\xleftarrow{\ana{b}}\crr b$.
\end{proposition}
\begin{proof}
Indeed, the definition of an anamorphism ensures the commutativity of the next diagram
\[\begin{diagram}
    \node{\crr r}
        \arrow{e,t}{r_b}
        \arrow{s,l}{r_a}
            \node{\crr b}
                \arrow{s,r}{\ana{b}} \\
    \node{\crr a}
        \arrow{e,t}{\ana{a}}
            \node{\crr{\fic{\Fnr{F}}}}
\end{diagram}    
\]
i.e. for the pullback $\crr a\xleftarrow{p_a}P\xrightarrow{p_b}\crr b$ of the cospan $\crr a\xrightarrow{\ana{a}}\crr{\fic{\Fnr{F}}}\xleftarrow{\ana{b}}\crr b$, we have the existence of a unique morphism $f:\crr r\to P$.
One can see that this morphism is a monomorphism taking into account that the mono-span and pullback properties. 
\end{proof}
This statement can be inverted for some class of endofunctors.
\begin{proposition}\label{prop:largest-bisimulation}
If category $\Cat{C}$ is finitely complete, endofunctor $\Fnr{F}$ preserves pullbacks, and there exists a final $\Fnr{F}$-coalgebra then the pullback of the cospan $\crr a\xrightarrow{\ana{a}}\crr{\fic{\Fnr{F}}}\xleftarrow{\ana{b}}\crr b$ is the greatest bisimulation of $\Fnr{F}$-coalgebras $a$ and $b$.
\end{proposition}
\begin{proof}
Indeed, let $\crr a\xleftarrow{p_a}P\xrightarrow{p_b}\crr b$ be the mentioned above pullback then it is mono-span and the next diagram
\[\begin{diagram}
    \node{\Fnr{F}P}
        \arrow{s,l}{\Fnr{F}p_a}
        \arrow{e,t}{\Fnr{F}p_b}
            \node{\Fnr{F\crr b}}
                \arrow{s,r}{\Fnr{F}\ana{b}} \\
    \node{\Fnr{F}\crr a}
        \arrow{e,t}{\Fnr{F}\ana{a}}
            \node{\Fnr{F}\crr{\fic{\Fnr{F}}}}
\end{diagram}
\]
is a pullback diagram.
Further, note that the outer square of the following diagram commutes.
\[\begin{diagram}
    \node{P}
        \arrow{se,tb}{\rho}{\exists}
        \arrow[2]{s,l}{p_a}
        \arrow[2]{e,t}{p_b}
            \node[2]{\crr b}
                \arrow{s,r}{b} \\
    \node[2]{\Fnr{F}P}
        \arrow{s,l}{\Fnr{F}p_a}
        \arrow{e,t}{\Fnr{F}p_b}
            \node{\Fnr{F\crr b}}
                \arrow{s,r}{\Fnr{F}\ana{b}} \\
    \node{\crr a}
        \arrow{e,t}{a}
            \node{\Fnr{F}\crr a}
                \arrow{e,t}{\Fnr{F}\ana{a}}
                    \node{\Fnr{F}\crr{\fic{\Fnr{F}}}}
\end{diagram}
\]
Indeed,
\begin{align*}
    \big(\Fnr{F}\ana{a}\,\big)ap_a&=(\fic{\Fnr{F}})\app{\ana{a}}{p_a} &&
        \text{considering that $\ana{a}$ is an $\Fnr{F}$-morphism} \\
    &=(\fic{\Fnr{F}})\app{\ana{b}}{p_b} &&
        \text{using the pullback property for $\crr a\xleftarrow{p_a}P\xrightarrow{p_b}\crr b$} \\
    &=\big(\Fnr{F}\ana{b}\,\big)bp_b &&
        \text{considering that $\ana{b}$ is an $\Fnr{F}$-morphism.}
\end{align*}
Thus, taking into account the pullback property for the inner subdiagram in the diagram above, one can derive the existence of $\rho$ ensuring the diagram commutativity.
But it means that the span $\crr a\xleftarrow{p_a}P\xrightarrow{p_b}\crr b$ can be lifted up to the corresponding span of coalgebras.
\end{proof}
Another concept used below is the concept of subcoalgebra.
\begin{definition}
Let $a\in\cat{Coalg}_{\Fnr{F}}(\Cat{C})$ and $j:\crr j\to\crr a$ be a monomorphism in $\Cat{C}$ with the domain $\crr j$ then an $\Fnr{F}$-coalgebra $c$ is called a \AT{subcoalgebra} of $a$ if $\crr c=\crr j$ and $j$ is lifted upto a coalgebraic monomorphism $j:c\to a$.
\end{definition}

We complete this section by enumerating a series of facts specific for the category $\cat{Set}$ and endofunctors of this category.
If $\Cat{C}=\cat{Set}$ then an $\Fnr{F}$-coalgebra is called an $\Fnr{F}$-\AT{system}\footnote{%
Due to J. Rutten \cite{bib:rutten}} and the corresponding category is denoted below by $\Sys{\Fnr{F}}$ instead of $\cat{Coalg}_{\Fnr{F}}(\cat{Set})$.

It is well-known that category $\cat{Set}$ is finitely complete.
The important class of endofunctors of the category $\cat{Set}$ is called the class of polynomial endofunctors and defined as follows (see, for example, \cite{bib:jacobs-2017})
\begin{enumerate}\item a constant endofunctor i.e. an endofunctor $\Fnr{C}X=\onm{C}$ for some set $\onm C$ and $\Fnr{C}f=\id[\onm C]$ is a polynomial endofunctor;
\item if $\Fnr{F}_1$ and $\Fnr{F}_2$ are polynomial endofunctors then the endofunctor $\Fnr{F}X=\Fnr{F}_1X\times\Fnr{F}_2X$ and $\Fnr{F}f=\Fnr{F}_1f\times\Fnr{F}_2f$ is a polynomial endofunctor;
\item if $\Fnr{F}_1$ and $\Fnr{F}_2$ are polynomial endofunctors then the endofunctor $\Fnr{F}X=\Fnr{F}_1X+\Fnr{F}_2X$ and $\Fnr{F}f=\Fnr{F}_1f+\Fnr{F}_2f$ is a polynomial endofunctor\footnote{%
    Here and below, ``$+$'' refers to a disjunctive sum of sets.}.
\end{enumerate}
An important property of polynomial endofunctor is that such an endofunctor preserves pullbacks.
%%
%% Systems as Coalgebras
\section{Systems as Coalgebras}\label{sec:systems}% file 14.tex
In this section, we introduce and study some endofunctors, systems related to the endofunctors, and their property.
Note that all considered here endofunctors are polynomial.
\subsection{Discrete Dynamical Systems}\label{subsec:T-system}
The simplest manner for modelling a discrete dynamical system is to represent it by a pair $(X,g)$ where $X$ is a set called the \AT{state set} and $g:X\to X$ is a mapping called the \AT{dynamics}.
\\
It is evident that $g$ can be considered as $\fnr[\cat{Set}]{Id}$-coalgebra of the category $\cat{Set}$ and, in this context, $X$ denoted by $\crr g$.
\\
Morphisms of such coalgebras are mappings that intertwine the dynamics of the corresponding coalgebras.
More precise, for $\fnr{Id}_{\cat{Set}}$-coalgebras $g$ and $h$, a morphism $f:g\to h$ is a mapping $f:\crr g\to\crr h$ such that $fg=hf$.

Easy seen that the final coalgebra exists in this case but it is trivial namely $\crr{\fic{\fnr{Id}_{\cat{Set}}}}=\onm1$ and $\fic{\fnr{Id}_{\cat{Set}}}=\id[\onm1]$.
Thus, taking into account that category $\cat{Set}$ and functor $\fnr{Id}_{\cat{Set}}$ satisfy conditions of Prop.~\ref{prop:largest-bisimulation} one can conclude that any two dynamical systems are bisimilar.

A less trivial example is given by the endofunctor $\fnr{T}:\cat{Set}\to\cat{Set}$ defined as follows\footnote{%
    Here and below, $\onm1$ refers to an one-element set $\{\fault\}$ containing a fault indicator.}
\begin{flushleft}
\begin{tabular}{clp{0.5\textwidth}}
    \hspace*{12pt} & $\fnr{T}X=\onm1+X$ & for any object $X$ of category $\cat{Set}$; \\
    \hspace*{12pt} & $\fnr{T}f=\ext f$ & for any morphism $f$ of category $\cat{Set}$.
\end{tabular}
\end{flushleft}
A system of such a kind is called a \AT{system with termination}.

The corresponding class of morphisms (see Def.~\ref{def:coalgebra-morphism}) from a $\fnr{T}$-system $g$ into a $\fnr{T}$-system $h$ contains a mapping $f:\crr g\to\crr h$ if and only if for any $x\in\crr g$, either $gx=\fault$ and $h(fx)=\fault$ or $gx\neq\fault$ and $f(gx)=h(fx)$.

There is a final $\fic{\fnr{T}}$ in the category $\Sys{\fnr{T}}$.
This object is structured as follows
\[\begin{array}{llcl}
    \crr{\fic{\fnr{T}}} &=1+\N & \text{where}& 1+\N=\N\bigcup\{\infty\}\\
    \fic{\fnr{T}} &= \labs{x\in1+\N}{\left\lbrace\begin{array}{cl}
        \fault & \text{if }x=0 \\
        \infty & \text{if }x=\infty \\
        x-1 & \text{otherwise}
    \end{array}\right.}
\end{array}
\]
For a $\fnr{T}$-system $g$, the corresponding anamorphism $\ana{g}$ is defined as follows
\[
    \ana{g}=\labs{x\in\crr g}{\left\lbrace\begin{array}{cl}
        \infty & \text{if }g^{(k)}x\neq\fault\text{ for all }k\in\N[+] \\
        \min\{k\in\N\mid g^{(k+1)}x=\fault\} & \text{otherwise}
    \end{array}\right.}
\]
where\footnote{%
Everybody familiar with the concept of a monad can easily see that $g^{(k)}$ is the $k^{\mathrm{th}}$-power of $g$ in the corresponding Kleisli category.}
\begin{eqnarray*}
    &g^{(1)}&=g \\
    &g^{(k+1)}&=\labs{x\in\crr g}{\left\lbrace\begin{array}{cl}
        \fault & \text{if }g^{(k)}x=\fault \\
        g(g^{(k)}x) & \text{otherwise} 
    \end{array}\right.}\qquad\text{for }k\in\N[+].
\end{eqnarray*}
\subsection{Systems with Output}\label{subsec:observed-systems}
In this subsection, we consider the class of dynamical systems equipped with a mechanism for the external monitoring of the current system state.
We use the term a \AT{system with output} for referring to a system of this class.
Following the coalgebraic approach, we model systems of this class as coalgebras, which signature is defined by the corresponding endofunctor $\S:\cat{Set}\to\cat{Set}$ that is defined as follows
\begin{flushleft}
\begin{tabular}{cll}
    \hspace*{12pt} & $\S X=\Ntf\times X$ & for any object $X$ of category $\cat{Set}$; \\
    \hspace*{12pt} & $\S f=\id[{\Ntf}]\times f$ &%
        for any objects $X$ and $Y$ of category $\cat{Set}$ and a morphism $f:X\to Y$.
\end{tabular}
\end{flushleft}
In this definition, $\Ntf$ refers to a finite set of possible output values.

Thus, a system with output is a mapping $\sigma:\crr\sigma\to\S\crr\sigma$ where $\crr\sigma$ is a set called usually the system \AT{state set}.
Taking into account the universality of the product (see, \cite[Sec.~III.4]{bib:maclane} or \cite[Sec.~2.4]{bib:awodey}), one can see that an $\S$-system $\sigma$ is uniquely represented as $\sigma=\tuple{\out{\sigma}}{\tr{\sigma}\rangle}$ where $\out\sigma=\pr_{\Ntf}\sigma:\crr\sigma\to\Ntf$ and $\tr\sigma=\pr_{\crr\sigma}\sigma:\crr\sigma\to\crr\sigma$ called respectively the \AT{output and transition functions} of the system.

The general coalgebraic framework (see Def.~\ref{def:coalgebra-morphism}) gives the following concept of an $\S$-morphism: for $\S$-system $\sigma$ and $\tau$, a mapping $f:\crr\sigma\to\crr\tau$ is called an $\S$-morphism if $(\out{\tau})f=\out{\sigma}$ and $(\tr{\tau})f=f(\tr{\sigma})$.

There is a final object $\fic{\S}$ in the category $\Sys{\S}$.
This object is structured as follows
\begin{flushleft}
\begin{tabular}{cll}
    \hspace*{12pt} & $\crr{\fic{\S}}$ & $=\Ntf^{\N}$; \\
    \hspace*{12pt} & $\out{(\fic{\S})}$ & $=\labs{s\in\Ntf^{\N}}{\app{s}{0}}$; \\
    \hspace*{12pt} & $\tr{(\fic{\S})}$ & $=\labs{s\in\Ntf^{\N}}{\labs{k\in\N}{\app{s}{(k+1)}}}$.
\end{tabular}
\end{flushleft}
For an $\Ntf$-system $\sigma$, the corresponding anamorphism $\ana{\sigma}$ is defined by the next formula
\[
    \ana{\sigma}=\labs{x\in\crr\sigma}{\labs{k\in\N}{{\out{\sigma}(\tr{\sigma}^kx)}}}.
\]

Thus, the proposed model allows considering a point of the final $\Ntf$-system carrier as an observed behaviour of the system being studied.
It means that we are interested in specifying and analysing constraints that allow distinguishing an admissible and inadmissible system behaviour.

\subsection{Detectors of Behavioural Violations}
The remaining part of the paper is devoted to demonstrating that the safeness of a subset can be described with using the category-theoretic language.

In this subsection, we introduce some class of systems that used as a tool for distinguishing admissible an inadmissible system behaviours.
We refer to systems of this class as detectors of behavioural violations.

This class is determined with the endofunctor $\D:\cat{Set}\to\cat{Set}$ defined as follows
\begin{flushleft}
\begin{tabular}{cll}
    \hspace*{12pt} & $\D X=(\onm1+X)^{\Ntf}$ & for any object $X$ of category $\cat{Set}$; \\
    \hspace*{12pt} & $\D f=\labs{\phi\in(\onm1+X)^{\Ntf}}{\big(\labs{n\in\Ntf}{(\ext{f})(\app{\phi}{n}})\big)}$%
        & for any objects $X$ and $Y$ of category $\cat{Set}$ and \\
    &&\hspace*{12pt}a morphism $f:X\to Y$.
\end{tabular}
\end{flushleft}

We call below a $\D$-system a \AT{detector} and for a detector $\dr a$, we refer to $\crr{\dr a}$ as to the detector \AT{state set}.

A detector morphism $f$ from a detector $\dr a$ into a detector $\dr b$ is (compare with Def.~\ref{def:coalgebra-morphism}) a mapping $f:\crr{\dr a}\to\crr{\dr b}$ such that for any $x\in\crr{\dr a}$ and $n\in\Ntf$,
\begin{enumerate}
    \item $\app{(\dr ax)}{n}=\fault$ if and only if $\app{(\dr b(fx))}{n}=\fault$;
    \item if $\app{(\dr ax)}{n}\neq\fault$ then
        $\app{\big(\dr b(fx)\big)}{n}=f\big(\app{(\dr ax)}{n}\big)$.
\end{enumerate}
\begin{proposition}
The class $\Sys{\D}$ of $\Ntf$-detectors equipped with detector morphisms forms a category.
\end{proposition}
\textit{Proof} is a simple check of the categories axioms.\qed

The following theorem describes a final detector.
\begin{theorem}\label{thm:final-detector}
There is a final object in category $\Sys{\D}$ that is structured as follows
\begin{flushleft}
\begin{tabular}{cll}
\hspace*{12pt} & $\crr{\fic{\D}}$ & is the set of all prefix-free subsets of $\Ntf^+$; \\
    \hspace*{12pt} & $\fic{\D}$%
        & $=\labs{P\in\crr{\fic{\D}}}{\labs{n\in\Ntf}{\left\lbrace\begin{array}{cl}
            \fault & \text{if }n\in P \\
            n^{-1}\cdot P & \text{otherwise}
        \end{array}\right.}}$.
\end{tabular}
\end{flushleft}
For an a $\Ntf$-detector, the corresponding anamorphism $\ana{\dr a}$ is defined by the next formula
\[
    \ana{\dr a}=\labs{x\in\crr{\dr a}}{%
        \{\bs{u}\in\Ntf^+\mid\dr a^+(x,\bs{u})=\fault\text{ and }\dr a^+(x,\bs{u}\slice{0}{k})\neq\fault%
        \text{ whenever }0<k<|\bs u|\}}
\]
where $\dr a^+:\crr{\dr a}\times\Ntf^+\to\onm1+\crr{\dr a}$ defined as follows
\begin{align*}
    \dr a^+(x,n)&=\app{(\dr ax)}{n} && x\in\crr{\dr a},\ n\in\Ntf; \\
    \dr a^+(x,\bs{u}n)&=\left\lbrace\begin{array}{cc}
        \fault & \text{if }\dr a^+(x,\bs{u})=\fault \\
        a(\dr a^+(x,\bs{u}),n) & \text{otherwise}
    \end{array}\right. && x\in\crr{\dr a}\,\ n\in\Ntf,\ \bs{u}\in\Ntf^+.
\end{align*}
\end{theorem}
The theorem can be derived from \cite[Lemma 6]{bib:jacobs-1996} but we give a direct proof, which details are used below.
The proof is preceded by a statement of the facts being used.
\begin{lemma}\label{lem:mult}
For any $\Ntf$-detector $\dr a$, $x\in\crr{\dr a}$, and $\bs{u},\bs{v}\in\Ntf^+$, the equation
\begin{equation}\label{eq:mult}
    \dr a^+(x,\bs{uv})=\left\lbrace\begin{array}{cc}
        \fault & \text{if }\dr a^+(x,\bs{u})=\fault \\
        \dr a^+\big(\dr a^+(x,\bs{u}),\bs{v}\big) & \text{otherwise} 
    \end{array}\right.
\end{equation}
holds.
\end{lemma}
\begin{proof}
If $\dr a^+(x,\bs u)=\fault$ then $\dr a^+(x,\bs{uv})=\fault$ by definition of $\dr a^+$.
Hence we below assume $\dr a^+(x,\bs u)\neq\fault$. 
\\
Let us use induction on $\bs v$.
\\
For $\bs{v}=n\in\Ntf$, \eqref{eq:mult} follows from the definition of $\dr a^+$.
\\
Assume that $\bs{v}=\bs{v}'n$, $n\in\Ntf$, and \eqref{eq:mult} holds for $\bs{v}'$.
\\
Then under assumption that $\dr a^+(x,\bs{uv}')=\fault$, we have $\dr a^+(x,\bs{u}(\bs{v}'n))=\dr a^+(x,(\bs{u}\bs{v}')n)=\fault$ by definition of $\dr a^+$.
Hence, $\dr a^+(x,\bs{uv})=\fault$.
\\
% In the other side, we have
% %
% \begin{align*}
%     \dr a^+(\dr a^+(x,\bs u),\bs v)&=\dr a^+(\dr a^+(x,\bs u),\bs{v}'n)
%         =\Big(\dr a\big(\dr a^+(\dr a^+(x,\bs u),\bs{v}')\big)\Big)n \\
%     &=\big(\dr a(\dr a^+(x,\bs{uv}')\big)n
%         && \text{by induction hypothesis} \\
%     &=\dr a^+\big(x,(\bs{uv}')n\big)=\dr a^+\big(x,\bs{u}(\bs{v}'n)\big)
%         =\dr a^+(x,\bs{uv})
% \end{align*}
%
In the other side, we have $\dr a^+(\dr a^+(x,\bs u),\bs {v}') = \dr a^+(x,\bs{uv}') = \fault$ by induction hypothesis.
But then  $\dr a^+(\dr a^+(x,\bs u),\bs v)=\dr a^+(\dr a^+(x,\bs u),\bs{v}'n) = \fault$  by definition of $\dr a^+$.
\\
Hence, we have $\dr a^+(x,\bs{uv}) =  \dr a^+(\dr a^+(x,\bs u),\bs v)$.
\\
In contrary, assumption that $\dr a^+(x,\bs{uv}')\neq\fault$ gives
\begin{align*}
    \dr a^+(x,\bs{uv})&=\dr a^+(x,\bs{u}(\bs{v}'n))=\dr a^+(x,(\bs{uv}')n)=\app{\big(\app{\dr a}{\dr a^+(x,\bs{uv}')}\big)}{n}
        && \text{by definition of }a^+ \\
    &=\app{\big(\app{\dr a}{\dr a^+(\dr a^+(x,\bs{u}),\bs{v}')}\big)}{n}  && \text{by induction hypothesis} \\
    &=\dr a^+(\dr a^+(x,\bs{u}),\bs{v}'n)=\dr a^+(\dr a^+(x,\bs{u}),\bs{v})
        && \text{by definition of }a^+.
\end{align*}
Thus, the lemma is proved.
\end{proof}
\begin{lemma}\label{lem:prefix-free}
If $P\subset\Ntf^+$ is prefix-free then $n^{-1}\cdot P\subset\Ntf^+$ and it is prefix-free whenever $n\in\Ntf$ and $n\notin P$.
\end{lemma}
\begin{proof}
Firstly, $n\notin P$ ensures $\nil\notin n^{-1}\cdot P$ i.e. $n^{-1}\cdot P\subset\Ntf^+$.
Further, if $\bs{u}\in n^{-1}\cdot P$ and $\bs{u}\slice{0}{m}\in n^{-1}\cdot P$ for $0\leq m<|\bs u|$ then $n\bs{u}\in P$ and $n\bs{u}\slice{0}{m}\in P$. Taking into account $n\bs{u}\slice{0}{m}=(n\bs{u})\slice{0}{m+1}$, one can obtain a contradiction, which completes the proof.
\end{proof}
\begin{lemma}\label{lem:dec}
Any prefix-free subset $P$ of $\Ntf^+$ can be represented as the following disjunctive union
\begin{equation}\label{eq:dec}
    P=\Ntf_P+\sum_{n\in\Ntf\setminus\Ntf_P}n\cdot P_n\qquad
        \text{where}\quad\Ntf_P=\{n\in\Ntf\mid n\in P\}\text{ and }P_n=n^{-1}\cdot P.
\end{equation}
\end{lemma}
\begin{proof}
Assume $\bs{u}\in P$ then either $\bs{u}=n$ for some $n\in\Ntf$ or $|\bs u|>1$.
In the first case, we have $\bs{u}\in\Ntf_P$ and, therefore, $\bs u$ belongs to the right side of \eqref{eq:dec}.
In another case, $\bs{u}=(\app{\bs u}{0})\bs{u}\slice{1}{}$ where $(\app{\bs u}{0})\notin\Ntf_P$ and $|\bs{u}\slice{1}{}|>0$ i.e. $\bs{u}\slice{1}{}\in(\app{\bs u}{0})^{-1}\cdot P$ and $\bs u$ belongs to the right side of \eqref{eq:dec}.
\\
Now assume $\bs u$ belongs to the right side of \eqref{eq:dec} then either $\bs{u}\in\Ntf_P$ or $\bs{u}\in\sum\limits_{n\in\Ntf\setminus\Ntf_P}n\cdot P_n$.
In the first case, we have $\bs{u}=n\in\Ntf_P$ i.e. $\bs{u}\in P$ by definition of $\Ntf_P$.
In another case, $\bs{u}\in n\cdot P_n$ where $n\in\Ntf\setminus\Ntf_P$.
Hence, $\bs{u}=n\bs{u}\slice{1}{}$ and $\bs{u}\slice{1}{}\in n^{-1}\cdot P$.
The last means that $\bs{u}\in P$.
\end{proof}
Now all is ready for proving Theorem~\ref{thm:final-detector}.
\begin{proof}[Proof of Theorem~\ref{thm:final-detector}]
Firstly, Lemma~\ref{lem:prefix-free} ensures that $\fic{\D}:\crr{\fic{\D}}\to\D\crr{\fic{\D}}$.
Hence, $\fic{\D}$ is an $\Ntf$-detector.
\\
Further, let us show that the mapping $\ana{\dr a}:\crr{\dr a}\to\crr{\fic{\D}}$ defined above for any $\Ntf$-detector $\dr a$ is an $\Ntf$-detector morphism.
\\
If $\app{(\dr ax)}{n}=\fault$ then $n\in\app{\ana{\dr a}}{x}$ i.e. $\big(\fic{\D}(\app{\ana{\dr a}}{x})\big)n=\fault$.
Conversely, if $\big(\fic{\D}(\app{\ana{\dr a}}{x})\big)n=\fault$ then $n\in\app{\ana{\dr  a}}{x}$ i.e. $\app{(\dr ax)}{n}=\fault$.
\\
If $\app{(\dr ax)}{n}\neq\fault$ then either $\app{\ana{\dr a}}{x}=\varnothing$ or this equation is wrong.
\\
In the first case, $\dr a^+(x,\bs{u})\neq\fault$ for any $\bs{u}\in\Ntf^+$ and, therefore, $\app{\ana{\dr a}}{\big(\app{(\dr ax)}{n}\big)}=\varnothing$ too.
Also, we have
\[
    \big(\fic{\D}(\,\app{\ana{\dr a}}{x})\big)n=(\fic{\D}\varnothing)n=n^{-1}\cdot\varnothing=\varnothing
        =\app{\ana{\dr a}}{\big(\app{(\dr ax)}{n}\big)}.
\]
In another case, we have both $\app{(\dr ax)}{n}\neq\fault$ and $\app{\ana{\dr a}}{x}\neq\varnothing$.
\\
If $\bs{u}\in\big(\fic{\D}(\,\app{\ana{\dr a}}{x})\big)n=n^{-1}\cdot(\,\app{\ana{\dr a}}{x})$ then $n\bs{u}\in\app{\ana{\dr a}}{x}$ i.e. $\dr a^+(x,n\bs{u})=\fault$ but $\dr a^+(x,n\bs{u}\slice{0}{m})\neq\fault$ whenever $m<|\bs u|$.
Hence, for any $m\leq|\bs u|$ and $\bs{v}=\bs{u}\slice{0}{m}$,  we have
\begin{equation}\label{eq:aux-final-detector}
    \dr a^+(x,n\bs{v})=\dr a^+(\dr a^+(x,n),\bs{v})=\dr a^+(\app{(\dr ax)}{n},\bs{v})
\end{equation}
due to Lemma~\ref{lem:mult} and by definition of $\dr a^+$.
Hence, one can conclude that $\bs{u}\in\app{\ana{\dr a}}{\big(\app{(\dr ax)}{n}\big)}$.
\\
Conversely, if $\bs{u}\in\app{\ana{\dr a}}{\big(\app{(\dr ax)}{n}\big)}$ then \eqref{eq:aux-final-detector} guarantees $\dr a^+(x,n\bs{u})=\fault$ but $\dr a^+(x,n\bs{u}\slice{0}{m})\neq\fault$ whenever $m<|\bs u|$.
It means that $\bs{u}\in n^{-1}\cdot(\,\app{\ana{\dr a}}{x})=\big(\fic{\D}(\,\app{\ana{\dr a}}{x})\big)n$.
\\
Thus, we have checked the $\Ntf$-detector morphism properties for $\ana{\dr a}$.
\\
For completing the proof, we need to show that $\ana{a}$ is the only $\Ntf$-detector morphism from $\dr a$ into $\fic{\D}$ i.e. $f=\ana{\dr a}$ for any $f:\dr a\to\fic{\D}$.
\\
Assume $fx=\varnothing$ then we have the next chain of equivalent statements
\begin{align*}
    &fx=\varnothing \\
    &n^{-1}\cdot(fx)=\varnothing\text{ for all }n\in\Ntf \\ 
    &\app{\big(\fic{\D}(fx)\big)}{n}=\varnothing\text{ for all }n\in\Ntf
        && \text{by definition of }\fic{\D}\\
    &\app{(\dr ax)}{n}\neq\fault\quad\text{and}\quad f\big(\app{(\dr ax)}{n}\big)=\varnothing\quad\text{for any }n\in\Ntf
        && \text{due to the $\Ntf$-morphism properties} \\
    &\dr a^+(x,\bs{u})\neq\fault\quad\text{for all }\bs{u}\in\Ntf^+ && \text{by induction on }\bs{u} \\
    &\app{\ana{\dr a}}{x}=\varnothing && \text{by definition of }\ana{\dr a}.
\end{align*}
Hence, $fx=\varnothing$ iff $\app{\ana{\dr a}}{x}=\varnothing$.
\\
Now using induction on $|\bs u|$, prove that $\bs{u}\in fx$ if and only if $\bs{u}\in\app{\ana{\dr a}}{x}$.
\\ % induction base
If $|\bs u|=1$ then we have the next chain of equivalent statements.
\begin{align*}
    &n\in fx\quad\text{for some }n\in\Ntf \\
    &\big(\fic{\D}(fx)\big)n=\fault && \text{due to Lemma~\ref{lem:dec}} \\
    &\app{(\dr ax)}{n}=\fault && \text{due to the $\Ntf$-morphism properties} \\
    &n\in\app{\ana{\dr a}}{x} && \text{by definition of }\ana{a}.
\end{align*}
Hence, $\bs{u}\in fx$ iff $\bs{u}\in\app{\ana{\dr a}}{x}$ for any $\bs{u}\in\Ntf^+$ such that $|\bs u|=1$.
\\ % induction step
If $|\bs u|=m>1$ and the required statement is true for $\bs{v}\in\Ntf^+$ such that $|\bs v|<m$ then we have the next chain of equivalent statements.
\begin{align*}
    &\bs{u}\in fx \\
    &\bs{u}\slice{1}{}\in(\app{\bs u}{0})^{-1}\cdot(fx) && \text{due to Lemma~\ref{lem:dec}} \\
    &\bs{u}\slice{1}{}\in\big(\fic{\D}(fx)\big)(\app{\bs u}{0}) && \text{by definition of }\fic{\D} \\
    &\bs{u}\slice{1}{}\in f\big(\app{(\dr ax)}{\app{(\bs u}{0})}\big)
        && \text{due to the $\Ntf$-morphism properties} \\
    &\bs{u}\slice{1}{}\in\app{\ana{\dr a}}{\big(\app{(\dr ax)}{(\app{\bs u}{0})}\big)}
        && \text{by induction hypothesis} \\
    &\dr a^+\big(\app{\dr (ax)}{(\app{\bs u}{0})},\bs{u}\slice{1}{}\big)=\fault\text{ but} \\
    &\qquad\dr a^+\big(\app{(\dr ax)}{(\app{\bs u}{0})},\bs{u}\slice{1}{k}\big)\neq\fault
        \text{ whenever }k<m && \text{by definition of }\ana{a} \\
    &\dr a^+\big(x,(\app{\bs u}{0})\bs{u}\slice{1}{}\big)=\fault\text{ but} \\
    &\qquad\dr a^+\big(x,(\app{\bs u}{0})\bs{u}\slice{1}{k}\big)\neq\fault
        \text{ whenever }k<m && \text{due to Lemma~\ref{lem:mult}} \\
    &\bs{u}\in\app{\ana{\dr a}}{x} && \text{by definition of }\ana{\dr a}.
\end{align*}
Thus, $f=\ana{\dr a}$\,.
\end{proof}
Below we need a characterisation of the subsets of $\crr{\fic{\D}}$ that are carriers of subcoalgebras.
The next proposition gives this characterisation.
\begin{proposition}\label{prop:subcoalgebra}
A subset $C\subset\crr{\fic{\D}}$ is the carrier of a subcoalgebras defined by the natural embedding $j_C:C\to\crr{\fic{\D}}$ if and only if the condition
\[
    \text{for any }P\in C\text{ and }n\in\Ntf,\quad\text{either}\quad n\in P\quad\text{or}\quad n^{-1}\cdot P\in C
\]
is fulfilled.
\end{proposition}
{\it Proof} boils down to simply checking the requirements of definitions.\qed

%%
%% Coalgebraic Framework for Understanding Behavioural Constraints
\section{Coalgebraic Understanding Safety Constraints}\label{sec:constraints}% file 15.tex
The general coalgebraic framework for recognising violations of the behaviour of a system with output is introduced and studied in this section.
\subsection{Functor Join and Its Properties}
This subsection defines bifunctor $\fnr{Join}$ (see Theorem~\ref{thm:join-bifuctor}), which is the key tool for our studying.
The importance of this bifunctor is related to the fact it preserves bisimulations (see Theorem~\ref{thm:bisimulation-preserving}).

Firstly assuming $\sigma$ and $\dr a$ is a system with output $\Ntf$ and an $\Ntf$-detector respectively, one can define the system with termination $\fnr{Join}(\sigma,\dr a):\crr\sigma\times\crr{\dr a}\to\fnr{T}(\crr\sigma\times\crr{\dr a})$ as follows
\begin{subequations}\label{eq:ver}
\begin{align}
    \fnr{Join}(\sigma,\dr a)&=\labs{(x,y)\in\crr\sigma\times\crr{\dr a}}{\left\lbrace\begin{array}{cc}
        \fault & \text{if }(\dr ay)(\app{\out{\sigma}}{x})=\fault \\
        \bigtuple{\app{\tr{\sigma}}{x}}{{(\dr ay)(\app{\out{\sigma}}{x})}} & \text{otherwise}
    \end{array}\right.}\label{eq:ver-1} \\
\intertext{%
Further for systems $\sigma$ and $\tau$ with output $\Ntf$ and $\Ntf$-detectors $\dr a$ and $\dr b$, an $\S$-morphism $f:\sigma\to\tau$, and a detector morphism $g:\dr a\to\dr b$, we define the mapping $\fnr{Join}(f,g):\crr\sigma\times\crr{\dr a}\to\crr\tau\times\crr{\dr b}$ by the formula}
    \fnr{Join}(f,g)&=f\times g.\label{eq:ver-2}
\end{align}
\end{subequations}
\begin{theorem}\label{thm:join-bifuctor}
Rules~\eqref{eq:ver} determine a bifunctor $\fnr{Join}:\Sys{\S}\times\Sys{\D}\to\Sys{\fnr{T}}$.
\end{theorem}
\begin{proof}
Firstly, let us check that for any $f:\sigma\to\tau$ and $g:\dr a\to\dr b$ where $\sigma,\tau\in\Sys{\S}$ and $a,b\in\Sys{\D}$, $\fnr{Join}(f,g)$ is a $\fnr{T}$-morphism from $\fnr{Join}(\sigma,\dr a)$ into $\fnr{Join}(\tau,\dr b)$. \\
Indeed, let $\big(\fnr{Join}(\sigma,\dr a)\big)(x,y)=\fault$ where $x\in\crr\sigma$, $y\in\crr{\dr a}$ then we have the next chain of equivalent statements
\begin{align*}
    &\big(\fnr{Join}(\sigma,\dr a)\big)\tuple{x}{y}=\fault \\
    &(\dr ay)(\out{\sigma}x)=\fault && \text{by definition of }\fnr{Join}(\sigma,\dr a) \\
    &(\dr ay)(\out{\tau}(fx))=\fault && \text{due to }f\text{ is an }\S\text{-morphism} \\
    &\big(\dr b(gy)\big)\big(\out{\tau}(fx)\big)=\fault
        && \text{due to }g\text{ is a }\D\text{-morphism} \\
    &\big(\fnr{Join}(\tau,\dr b)\big)\tuple{fx}{gy}=\fault
        && \text{by definition of }\fnr{Join}(\tau,\dr b) \\
    &\big(\fnr{Join}(\tau,\dr b)\big)\big(\fnr{Join}(f,g)\tuple{x}{y}\big)=\fault
        && \text{by definition of }\fnr{Join}(f,g)
\end{align*}
Now assume that $\big(\fnr{Join}(\sigma,\dr a)\big)\tuple{x}{y}\neq\fault$\footnote{%
It means $(\dr ay)(\out\sigma x)\neq\fault$} where $x\in\crr\sigma$, $y\in\crr{\dr a}$ then
\begin{align*}
    \big(\fnr{Join}(f,g)\big)\Big(\big(\fnr{Join}(\sigma,\dr a)\big)\tuple{x}{y}\Big)
        &=\big(\fnr{Join}(f,g)\big)\bigtuple{\tr{\sigma}x}{g\big((\dr ay)(\out{\sigma}x)\big)}
        && \text{by definition of }\fnr{Join}(\sigma,\dr a) \\
    &=\bigtuple{f(\tr{\sigma}x)}{g\big((\dr ay)(\out{\sigma}x)\big)}
        && \text{by definition of }\fnr{Join}(f,g) \\
    &=\bigtuple{\tr{\tau}(fx)}{g\big((\dr ay)\big(\out{\tau}(fx)\big)}
        && \text{considering that $f$ is an $\S$-morphism} \\
    &=\bigtuple{\tr{\tau}(fx)}{\big(\dr b(gy)\big)\big(\out{\tau}(fx)\big)}
        && \text{considering that $g$ is a $\D$-morphism} \\
    &=\big(\fnr{Join}(\tau,\dr b)\big)\bigtuple{fx}{gy}
        && \text{by definition of }\fnr{Join}(\tau,\dr b) \\
    &=\big(\fnr{Join}(\tau,\dr b)\big)\big(\fnr{Join}(f,g)\tuple{x}{y}\big)
        && \text{by definition of }\fnr{Join}(f,g).
\end{align*}
Thus, $\fnr{Join}(f,g)$ is a $\fnr{T}$-morphism (see Subsec.~\ref{subsec:T-system}).
\\
Taking into account rule \eqref{eq:ver-2} one can conclude that $\fnr{Join}$ is a bifunctor.
\end{proof}
\begin{theorem}\label{thm:bisimulation-preserving}
Let $\sigma\xleftarrow{p_\sigma}\rho\xrightarrow{p_\tau}\tau$ be a bisimulation of systems $\sigma$ and $\tau$ with output $\Ntf$ and $a\xleftarrow{q_{\dr a}}\dr r\xrightarrow{q_{\dr b}}b$ be a bisimulation of $\Ntf$-detectors $\dr a$ and $\dr b$ then $\fnr{Join}(\rho,\dr r)$ is a bisimulation of $\fnr{Join}(\sigma,\dr a)$ and $\fnr{Join}(\tau,\dr b)$.
\end{theorem}
\begin{proof}
Let us consider the span $\fnr{Join}(\sigma,\dr a)\xleftarrow{\fnr{Join}(p_\sigma,q_{\dr a})}\fnr{Join}(\rho,\dr r)\xrightarrow{\fnr{Join}(p_\tau,q_{\dr b})}\fnr{Join}(\tau,\dr b)$. It is sufficient to show that the corresponding span $\crr\sigma\times\crr{\dr a}\xleftarrow{p_\sigma\times q_{\dr a}}\crr\rho\times\crr{\dr r}\xrightarrow{p_\tau\times q_{\dr b}}\crr\tau\times\crr{\dr b}$ in \cat{Set} is a mono-span. \\
Assume some mappings $g,h:S\to\crr\rho\times\crr{\dr r}$ with common domain $S$ satisfy the equations $(p_\sigma\times q_{\dr a})g = (p_\sigma\times q_{\dr a})h$ and $(p_\tau\times q_{\dr b})g = (p_\tau\times q_{\dr b})h$.
These mapping can be uniquely represented as follows $g=\langle g_{\crr\rho},g_{\crr{\dr r}}\rangle$ and $h=\langle h_{\crr\rho},h_{\crr{\dr r}}\rangle$ respectively where $g_{\crr\rho},h_{\crr\rho}:S\to\crr\rho$ and $g_{\crr{\dr r}},h_{\crr{\dr r}}:S\to\crr{\dr r}$. \\
Taking into account that %
\[
    (p_\sigma\times q_{\dr a})g =
        (p_\sigma\times q_{\dr a})\langle g_{\crr\rho},g_{\crr{\dr r}}\rangle =
        \langle p_\sigma g_{\crr\rho},q_{\dr a} g_{\crr{\dr r}}\rangle
        \qquad\text{and}\qquad
    (p_\sigma\times q_{\dr a})h =
        (p_\sigma\times q_{\dr a})\langle h_{\crr\rho},h_{\crr{\dr r}}\rangle =
        \langle p_\sigma h_{\crr\rho},q_{\dr a} h_{\crr{\dr r}}\rangle
\]
we have $\langle p_\sigma g_{\crr\rho},q_{\dr a} g_{\crr{\dr r}}\rangle = \langle p_\sigma h_{\crr\rho},q_{\dr a} h_{\crr{\dr r}}\rangle$ i.e. $p_\sigma g_{\crr\rho} = p_\sigma h_{\crr\rho}$ and $q_{\dr a} g_{\crr{\dr r}} = q_{\dr a} h_{\crr{\dr r}}$ due to properties of product.
Similarly, we have $p_\tau g_{\crr\rho} = p_\tau h_{\crr\rho}$ and $q_{\dr b} g_{\crr{\dr r}} = q_{\dr b} h_{\crr r}$ from condition $(p_\tau\times q_{\dr b})g = (p_\tau\times q_{\dr b})h$. \\
Then one can derive $g_{\crr\rho} = h_{\crr\rho}$ using the conditions $p_\sigma g_{\crr\rho} = p_\sigma h_{\crr\rho}$ and $p_\tau g_{\crr\rho} = p_\tau h_{\crr\rho}$ and the bisimulation $\sigma\xleftarrow{p_\sigma}\rho\xrightarrow{p_\tau}\tau$. \\
Similarly, we have $g_{\crr{\dr r}} =  h_{\crr{\dr r}}$ due to the conditions $q_{\dr a} g_{\crr{\dr r}} = q_{\dr a} h_{\crr{\dr r}}$ and $q_{\dr b} g_{\crr{\dr r}} = q_{\dr b} h_{\crr{\dr r}}$ and the bisimulation $\dr a\xleftarrow{q_{\dr a}}\dr r\xrightarrow{q_{\dr b}}\dr b$. \\
Thus, $g = h$ and, therefore, the considering span $\crr\sigma\times\crr{\dr a}\xleftarrow{p_\sigma\times q_{\dr a}}\crr\rho\times\crr{\dr r}\xrightarrow{p_\tau\times q_{\dr b}}\crr\tau\times\crr{\dr b}$ is realy a mono-span.
\end{proof}
\subsection{Safety Constraints and Detectors}
This subsection is central to the paper.
Here we establish an association between $\Ntf$-detectors and some subsets of $\Ntf^{\N}$.
Further, we prove this class of subsets is exactly the class of safety constraints.

First of all for $\bs{s}\in\Ntf^{\N}$, let us define the following system $[\bs s]$ with output namely
\[
    \crr{[\bs s]}=\big\{\bs{s}\slice{k}{}\mid k\in\N\big\}\quad\text{and}\quad
    [\bs s]=\labs{\bs{t}\in\crr{[\bs s]}}{\bigtuple{\app{\bs t}{0}}{\bs{t}\slice{1}{}}}.
\]
Further for any $\Ntf$-detector $\dr a$ and $x\in\crr{\dr a}$, let us define the following set
\[
    \constr{\dr a}{x}=\{\bs{s}\in\Ntf^{\N}\mid\app{\ana{\fnr{Join}([\bs s],\dr a)}}{\tuple{\bs s}{x}}=\infty\}.
\]
The next simple fact is useful below.
\begin{lemma}\label{lem:basic}
For $\bs{s}\in\Ntf^{\N}$, an $\Ntf$-detector $\dr a$, and $x\in\crr{\dr a}$,
\[
    \big(\fnr{Join}([\bs s],\dr a)\big)^{(m)}\tuple{\bs s}{x}=\left\lbrace\begin{array}{cc}
        \fault & \text{if }\dr a^+(x,\bs{s}\slice{0}{m})=\fault \\
        \tuple{\bs{s}\slice{m}{}}{\dr a^+(x,\bs{s}\slice{0}{m})} & \text{otherwise} 
    \end{array}\right.\quad\text{for }m>0.
\]
\end{lemma}
\begin{proof}
For proving we apply induction on $m$.
If $m=1$ then the equation holds due to \eqref{eq:ver-1}.
\\
Now assume the equation holds for some $m>1$.
\\
Let $\big(\fnr{Join}([\bs s],\dr a)\big)^{(m)}\tuple{\bs s}{x}=\fault$ then $\dr a^+(x,\bs{s}\slice{0}{m})=\fault$ by induction hypothesis.
The definition of $\dr a^+$ ensures $\dr a^+(x,\bs{s}\slice{0}{m+1})=\fault$.
But $\big(\fnr{Join}([\bs s],\dr a)\big)^{(m+1)}\tuple{\bs s}{x}=\fault$ by definition.
Hence, the equation holds in this case.
\\
Finally, assume $\big(\fnr{Join}([\bs s],\dr a)\big)^{(m)}\tuple{\bs s}{x}\neq\fault$ then $\big(\fnr{Join}([\bs s],\dr a)\big)^{(m)}\tuple{\bs s}{x}=\tuple{\bs{s}\slice{m}{}}{\dr a^+(x,\bs{s}\slice{0}{m})}$ by induction hypo\-thesis i.e. $\dr a^+(x,\bs{s}\slice{0}{m})\neq\fault$.
Thus,
\[
    \big(\fnr{Join}([\bs s],\dr a)\big)^{(m+1)}\tuple{\bs s}{x}
        =\big(\fnr{Join}([\bs s],\dr a)\big)\Big(\big(\fnr{Join}([\bs s],\dr a)\big)^{(m)}\tuple{\bs s}{x}\Big)
        =\big(\fnr{Join}([\bs s],\dr a)\big)\tuple{\bs{s}\slice{m}{}}{\dr a^+(x,\bs{s}\slice{0}{m})}.
\]
Now applying \eqref{eq:ver-1}, we obtain the required expression for $\big(\fnr{Join}([\bs s],\dr a)\big)^{(m+1)}\tuple{\bs s}{x}$.
\end{proof}
\begin{lemma}\label{lem:detector2safe}
For any $\Ntf$-detector $\dr a$ and $x\in\crr{\dr a}$, the set $\constr{\dr a}{x}$ is a safety constraint.
\end{lemma}
\begin{proof}
Indeed, let us assume the existence of $\bs{s}\notin\constr{\dr a}{x}$ such that the equation $\bs{s}'\slice{0}{m}=\bs{s}\slice[*]{0}{m}$ holds for any $m\in\N$ and for some $\bs{s}'\in\constr{\dr a}{x}$ depending in generally on $m$.
\\
The fact $\bs{s}\notin\constr{\dr a}{x}$ ensures $\ana{\fnr{Join}([\bs s],\dr  a)}\tuple{\bs s}{x}= K$ for some $K\in\N$.
It means (see Lemma~\ref{lem:basic}) $\dr a^+\big(x,\bs{s}\slice{0}{K+1}\big)=\fault$.
Let $\bs{s}'\in\constr{\dr a}{x}$ such that $\bs{s}'\slice[*]{0}{K+1}=\bs{s}\slice[*]{0}{K+1}$ then $\dr a^+\big(x,\bs{s}'\slice{0}{K+1}\big)=\fault$ and, therefore, $\ana{\fnr{Join}([\bs s'],\dr a)}\tuple{\bs s'}{x}=K$.
But $\ana{\fnr{Join}([\bs s'],\dr a)}\tuple{\bs s'}{x}=\infty$ by the assumption $\bs{s}'\in\constr{\dr a}{x}$.
This contradiction completes the proof.
\end{proof}
The following lemma is less simple.
\begin{lemma}\label{lem:safe2detector}
For any safety constraint $S\subset\Ntf^{\N}$, there exist an $\Ntf$-detector $\dr a_S$ and $x\in\crr{\dr a_S}$ such that $\constr{\dr a_S}{x}=S$.
\end{lemma}
\begin{proof}
Let us define the $\Ntf$-detector $\dr a_S$ as follows\footnote{%
Note that $\nil\in\crr{\dr a_S}$.}
\[\begin{split}
    \crr{\dr a_S}&=\{\nil\}\bigcup\big\{\bs{u}\in\Ntf^+\mid\bs{u}=\bs{s}\slice{0}{|\bs u|}%
        \text{ for some }\bs{s}\in S\big\} \\
    \dr a_S&=\labs{\bs{u}\in\crr{\dr a_S}}{\labs{n\in\Ntf}{\left\lbrace\begin{array}{cc}
        \fault & \text{if }\bs{u}n\notin\crr{a_S} \\
        \bs{u}n & \text{otherwise}
    \end{array}\right.}}
\end{split}
\]
and prove that $\constr{\dr a_S}{\nil}=S$.
\\
Let us assume $\bs{s}\in S$ then $\bs{s}\slice{0}{m}\in\crr{\dr a_S}$ for any $m\in\N$.
Hence, $\dr a_S^+(\bs{s}\slice{0}{m},\app{\bs{s}}{m})=\bs{s}\slice{0}{m+1}\in\crr{\dr a_S}$ for each $m$ i.e. $\ana{\fnr{Join}([\bs s],\dr a_S)}\tuple{\bs s}{\nil}=\infty$.
Thus, $\bs{s}\in\constr{\dr a_S}{\nil}$.
\\
Conversely assume $\bs{s}\in\constr{\dr a_S}{\nil}$ then $\ana{\fnr{Join}([\bs s],\dr a_S)}\tuple{\bs s}{\nil}=\infty$ and, therefore, $\big(\fnr{Join}([\bs s],\dr a_S)\big)^{(m)}\tuple{\bs s}{\nil}\neq\fault$ for each $m>0$. 
Lemma~\ref{lem:basic} ensures $\bs{s}\slice{0}{m}=\dr a_S^+(\nil,\bs{s}\slice{0}{m})\in\crr{\dr a_S}$ for all $m>0$ that guarantees existence $\bs{s}^{(m)}\in S$ such that $\bs{s}^{(m)}\slice{0}{m}=\bs{s}\slice{0}{m}$.
Now using the safeness of $S$, one can conclude $\bs{s}\in S$.
\end{proof}
\begin{theorem}[about universal detector]\label{thm:safety-coalgebraically}
A subset $S\subset\Ntf^{\N}$ is a safety constraint if and only if there exist $P\in\crr{\fic{\D}}$ such that  $S=\constr{\fic{\D}}{\strut P}$.
\end{theorem}
\begin{proof}
Lemmas~\ref{lem:detector2safe} and \ref{lem:safe2detector} ensure that a subset $S\subset\Ntf^{\N}$ is a safety constraint if and only if there exist an $\Ntf$-detector $\dr a$ and $x\in\crr{\dr a}$ such that $S=\constr{\dr a}{x}$.
Let us take $P=\app{\ana{\dr a}}{x}$ and prove $\ana{\fnr{Join}([\bs s],\dr a)}\tuple{\bs s}{x}=\ana{\fnr{Join}([\bs s],\fic{\D})}\tuple{\bs s}{P}$.
\\
Indeed, for the $\fnr{T}$-morphism $\fnr{Join}\big(\id[\crr{[\bs s]}],\ana{\dr a}\,\big):\fnr{Join}([\bs s],\dr a)\to\fnr{Join}([\bs s],\fic{\D})$, we have
\begin{equation}\label{eq:udet}
    \Big(\fnr{Join}\big(\id[\crr{[\bs s]}],\ana{\dr a}\,\big)\Big)\tuple{\bs s}{x}=
        \big(\id[\crr{[\bs s]}]\times\,\ana{\dr a}\,\big)\tuple{\bs s}{x}=\tuple{\bs s}{P}.
\end{equation}
Using the equation
\[
    \ana{\fnr{Join}([\bs s],\dr a)}=\ana{\fnr{Join}([\bs s],\fic{\D})}\circ
        \Big(\fnr{Join}\big(\id[\crr{[\bs s]}],\ana{\dr a}\,\big)\Big)
\]
that follows from the definition of an anamorphism, one can derive from \eqref{eq:udet} the follows
\[
    \ana{\fnr{Join}([\bs s],\dr a)}\tuple{\bs s}{x}=\ana{\fnr{Join}([\bs s],\fic{\D})}\tuple{\bs s}{P}.
\]
Considering the last equation holds whenever $P=\app{\ana{\dr a}}{x}$, one can conclude that $S=\constr{\dr a}{x}=\constr{\fic{\D}}{\strut P}$.
\end{proof}
\begin{corollary}
A subset $S\subset\Ntf^{\N}$ is a safety constraint if and only if it is equal to $\{\bs{s}\in\Ntf^{\N}\mid\bs{s}\slice{0}{m}\notin P\text{ for any }m>0\}$ for some prefix-free subset of $\Ntf^+$.
\end{corollary}
\begin{proof}
It follows immediately from Theorem~\ref{thm:universality}.
\end{proof}
\begin{corollary}
For any $\Ntf$-detectors $\dr a$ and $\dr b$ and $x\in\crr{\dr a}$ and $y\in\crr{\dr b}$, the equation  $\app{\ana{\dr a}}{x}=\app{\ana{\dr b}}{y}$ is sufficient for $\constr{\dr a}{x}=\constr{\dr b}{y}$.
\end{corollary}
\begin{proof}
The statement is true due to the construction method of $P$ used in the proof of Theorem~\ref{thm:safety-coalgebraically} and Theorem~\ref{thm:universality}.
\end{proof}
\subsection{Families of safety Constraints}
Theorem~\ref{thm:safety-coalgebraically} proved in the previous subsection establishes that for the family of all safety constraints, there is a universal detector, i.e. a detector that recognises any safety constraint of the family when the detector configured appropriately.
Such a general result, however, cannot be used in the practice of developing software tools, if only because computability considerations limit our expressive capabilities for specifying safety constraints and, in particular, guarantee the impossibility of specifying an arbitrary prefix-free set.
Therefore, we need to consider more specific families of safety constraints.
In this subsection, we try to outline some general approach to solving this problem.

We begin with the next definition.
\begin{definition}
A family of safety constraints $\mathcal{F}$ is below called a \AT{family with a universal detector} if $\mathcal F=\{\constr{\dr a}{x}\mid x\in\crr{\dr a}\}$ for some $\Ntf$-detector $\dr a$  being called in this case a \AT{universal detector} for $\mathcal F$.
\end{definition}
\begin{theorem}\label{thm:universality}
A family of safety constraints $\mathcal{F}$ is a family with a universal detector if and only if there exists $C\subset\crr{\fic{\D}}$ that meets the following conditions
\begin{subequations}\label{eq:universality}
\begin{align}
    &\text{for any }P\in C,\quad%
        \text{either}\quad n\in P\quad\text{or}\quad%
        n^{-1}\cdot P\in C\quad\text{for an arbitrary }n\in\Ntf\label{eq:universality-1} \\
    &\text{for any safety constraint }S,\quad%
        S\in\mathcal F\quad\text{if and only if there exists}\quad P\in C\quad\text{such that}\quad%
        S=\constr{\fic{\D}}{\strut P}\label{eq:universality-2}
\end{align}
\end{subequations}
\end{theorem}
\begin{proof}
Firstly, let us assume the family $\mathcal F$ is a family with a universal detector then Prop.~\ref{prop:subcoalgebra} guarantees the validity of conditions \eqref{eq:universality}.
\\
Conversely, let us define the next mapping $\dr a:C\to\D C$
\[
    \dr a=\labs{P\in C}{\labs{n\in\Ntf}{\left\lbrace\begin{array}{cc}
        \fault & \text{if }n\in P \\
        n^{-1}\cdot P & \text{otherwise}
    \end{array}\right.}}
\]
The correctness of this definition is ensured by \eqref{eq:universality-1}, hence $\dr a$ is an $\Ntf$-detector and, due to \eqref{eq:universality-2}, it is a universal detector.
\end{proof}
\begin{corollary}
A family of safety constraints $\mathcal{F}$ is a family with a universal detector if and only if
\[
    \{P\in\fic{\D}\mid\constr{\fic{\D}}{P}\in\mathcal F\}
\]
is the carrier of a subcoalgebra of $\fic{\D}$.
\end{corollary}
Now we consider three specific safety constraint family and demonstrate that each of them is a family with a universal detector.
\subsubsection{Regular Safety Constraints}
Here we consider the safety constraint family $\mathcal R$ formed as follows a safety constraint $S$ belongs to $\mathcal R$ if and only if there exists a detector $\dr a$ with the finite carrier such that $S=\constr{\dr a}{x}$ for some $x\in\crr{\dr a}$.
We call a safety constraint belonging this family a \AT{regular safety constraint}.
\begin{theorem}
The family of regular safety constraints is a family with a universal detector.
\end{theorem}
\begin{proof}
Let us consider $C=\{P\subset\Ntf^+\mid P=\app{\ana{a}}{x}\text{ for some detector }\dr a%
\text{ with finite carrier and }x\in\crr{\dr a}\}$ then any $P\in C$ is a regular prefix-free subset of $\Ntf^+$ (see, for example, \cite[Theorem~1]{bib:zholtkevych-polyakovska}).
\\
Now let us check that $C$ satisfies conditions \eqref{eq:universality}.
Indeed, Lemma~\ref{lem:prefix-free} ensures $n^{-1}\cdot P$ is prefix-free for any $n\in\Ntf$ such that $n\notin P$.
\\
Hence, we need only due to Theorem~\ref{thm:universality} to check that $n^{-1}\cdot P$ is a regular set under assumption $n\notin P$.
To do this we consider a finite Eilenberg machine $\langle Q,T\subset Q\times\Ntf\times Q,I\subset Q,F\subset Q\rangle$ that accepts only words from $P$ (see \cite{bib:eilenberg}) and construct the new finite Eilenberg machine $\langle Q',T'\subset Q'\times\Ntf\times Q',I'\subset Q',F\subset Q'\rangle$ where $Q'=Q+\{q^\ast\}$ and $T'=T+\{\langle q^\ast,n,q\rangle\mid\langle q',n,q\rangle\text{ for some }q'\in I\}$ and $q^\ast\notin Q$.
It is evident that the constructed machine accepts exactly $n^{-1}\cdot P$.
\end{proof}
\subsubsection{Decidable Safety Constraints}
Here we consider the safety constraint family $\mathcal D$ formed as follows a safety constraint $S$ belongs to $\mathcal D$ if and only if there exists $P\in\crr{\fic{\D}}$ such that $P$  is decidable and $S=\constr{\fic{\D}}{\strut P}$.
We call a safety constraint belonging this family a \AT{decidable safety constraint}.
\begin{theorem}
The family of decidable safety constraints is a family with a universal detector.
\end{theorem}
\begin{proof}
To prove the theorem is sufficient to check that the family of a decidable prefix-free subset of $\Ntf^+$ satisfies condition \eqref{eq:universality-2}.
But this is really true because one can check the statement $\bs{u}\in n^{-1}\cdot P$ for a decidable prefix-free subset $P$ of $\Ntf^+$, any $\bs{u}\in\Ntf^+$, and $n\in\Ntf$ such that $n\notin P$, by checking $n\bs{u}\in P$ with a decision procedure for $P$.
\end{proof}
\subsubsection{Recursively Enumerable Safety Constraints}
Here we consider the safety constraint family $\mathcal{RE}$ formed as follows a safety constraint $S$ belongs to $\mathcal{RE}$ if and only if there exists $P\in\crr{\fic{\D}}$ such that $P$  is recursively enumerable and $S=\constr{\fic{\D}}{\strut P}$.
We call a safety constraint belonging this family a \AT{recursively enumerable safety constraint}.
\begin{theorem}
The family of recursively enumerable safety constraints is a family with a universal detector.
\end{theorem}
\begin{proof}
Similarly to proof of the previous theorem, we need to furnish a semi-decision procedure for checking the statement $\bs{u}\in n^{-1}\cdot P$ with a recursively enumerable prefix-free subset $P$ of $\Ntf^+$, any $\bs{u}\in\Ntf^+$, and $n\in\Ntf$ such that $n\notin P$.
This procedure is as follows
\begin{enumerate}
    \item run parallelly checkings $n\bs{u}\slice{0}{k}\in P$ for $k=1,\ldots,|\bs u|$;
    \item wait for any of these runs to halt;
    \item report success if the halting run is the run for checking $n\bs{u}\in P$.
\end{enumerate}
An event waited in item~(2) may do not happen if $n\bs{u}\notin P$ i.e. $\bs{u}\notin n^{-1}\cdot P$.
\\
If the wait in paragraph (2) is completed then only for one $0<k<|\bs u|$ the corresponding checking is completed due to $P$ is prefix-free.
\\
$\bs{u}\in n^{-1}\cdot P$ or equivalently $n\bs{u}\in P$ if and only if the checking with number $|\bs u|-1$ is halted.
\\
In other cases, $n\bs{u}\notin P$ i.e. $\bs{u}\notin n^{-1}\cdot P$.
\\
Thus, the furnished procedure is really a semi-decision procedure for the statement $\bs{u}\in n^{-1}\cdot P$.
\end{proof}

\section{Conclusion}
Summing up the mentioned above we can conclude that subdetectors of a final detector correspond to families of safety constraints with universal detectors.
These families are candidates for semantic models of domain-specific languages for specification of safety constraints.

Of course, the obtained results tell nothing how to construct such languages.
But they turn this problem into more precisely defined.

We can identify as problems for further studying the follows
\begin{enumerate}
    \item Give syntactic characterisation of regular safety constraints family.
    \item Understand what classes of grammars (for example, context-free) define proper families of safety constraints.
    \item Can we use complexity classes for defining families of safety constraints with universal detectors?
    \item Understand how compositional theory of detectors can be developed.
\end{enumerate}

An especial area of further research is the dissemination of the proposed approach for the specification and analysis of causality constraints, formulated in terms of the logical clock model.
Now authors are working on the paper ``Understanding Logical Clock Model Coalgebraically''.
\bibliographystyle{unsrt}
\bibliography{references}
\end{document}